\documentclass[a4paper,UKenglish,cleveref, autoref, thm-restate]{lipics_v2021}

\pdfoutput=1 
\hideLIPIcs  

\title{Improved Approximations for Translational Packing of Convex Polygons} 

\author{Adam Kurpisz}{Bern University of Applied Sciences and Department of Mathematics, ETH Zurich, Switzerland \and \url{https://n.ethz.ch/~kurpisza/} }{adam.kurpisz@ifor.math.ethz.ch}{https://orcid.org/0000-0002-2320-4482}{}

\author{Silvan Suter}{Department of Mathematics, ETH Zurich, Switzerland}{sili.suter@bluewin.ch}{[orcid]}{}

\authorrunning{A. Kurpisz and S. Suter} 

\Copyright{Adam Kurpisz and Silvan Suter} 



\category{} 

\relatedversion{This is the full version of the same-named ESA 2023 paper. A link will be provided in a future version of this document.} 

\nolinenumbers 

\newtheorem{problem}[theorem]{Problem}

\usepackage{thmtools}
\usepackage{thm-restate}

\DeclareMathOperator{\dist}{dist}
\DeclareMathOperator{\opt}{opt}

\DeclareMathOperator{\argmin}{argmin}
\DeclareMathOperator{\argmax}{argmax}
\DeclareMathOperator{\Rectim}{\hspace{0.05 cm} \includegraphics[scale=0.02]{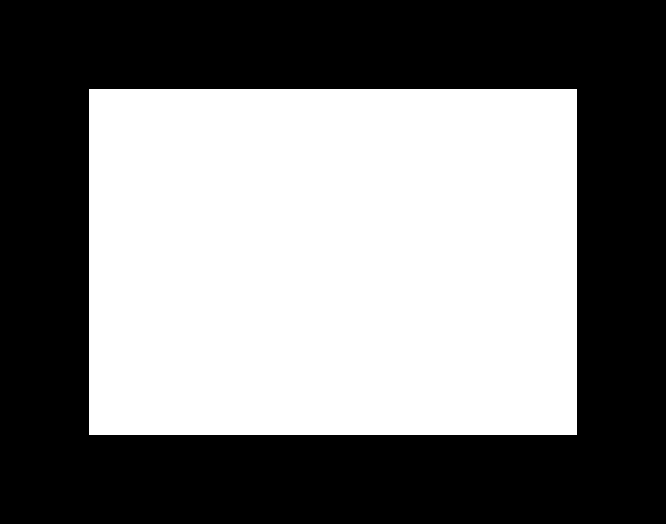}}
\DeclareMathOperator{\rectangles}{\mathcal{I}_{\Rectim}}
\DeclareMathOperator{\Polyim}{\hspace{0.02 cm} \includegraphics[scale=0.03]{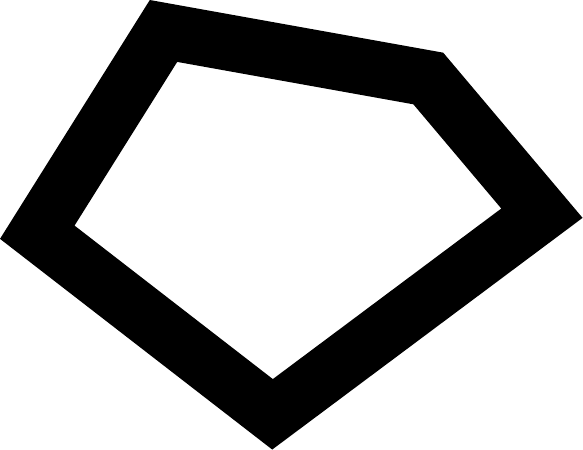}}
\DeclareMathOperator{\polygons}{\mathcal{I}_{\Polyim}}
\DeclareMathOperator{\area}{\mathtt{area}}
\DeclareMathOperator{\width}{\mathtt{width}}
\DeclareMathOperator{\height}{\mathtt{height}}
\DeclareMathOperator{\base}{\mathtt{base}}
\DeclareMathOperator{\wside}{\mathtt{wside}}

\DeclareMathOperator{\R}{\mathbb{R}}
\DeclareMathOperator{\N}{\mathbb{N}}
\DeclareMathOperator{\Q}{\mathbb{Q}}

\begin{document}

\maketitle

\begin{abstract}
Optimal packing of objects in containers is a critical problem in various real-life and industrial applications. This paper investigates the two-dimensional packing of convex polygons without rotations, where only translations are allowed. We study different settings depending on the type of containers used, including minimizing the number of containers or the size of the container based on an objective function.

Building on prior research in the field, we develop polynomial-time algorithms with improved approximation guarantees upon the best-known results by Alt, de Berg and Knauer, as well as Aamand, Abrahamsen, Beretta and Kleist, for problems such as Polygon Area Minimization, Polygon Perimeter Minimization, Polygon Strip Packing, and Polygon Bin Packing. Our approach utilizes a sequence of object transformations that allows sorting by height and orientation, thus enhancing the effectiveness of shelf packing algorithms for polygon packing problems. In addition, we present efficient approximation algorithms for special cases of the Polygon Bin Packing problem, progressing toward solving an open question concerning an $\mathcal{O}(1)$-approximation algorithm for arbitrary polygons.
\end{abstract}

\newpage
\section{Introduction}



Many real-life situations require us to make decisions about optimally packing a collection of objects into a specific container. One particular category of these packing problems is two-dimensional packing, which is encountered in everyday scenarios like arranging items on a shelf and in industrial applications such as cutting cookies from rolled-out dough or manufacturing sets of tiles from standard-sized panels made of wood, glass, or metal. Another intriguing example involves cutting fabric pieces for clothing production. In this case, the pieces often cannot be rotated freely, as they must adhere to a desired pattern in the final product, which is tailored of multiple elements.
The widespread applicability of two-dimensional packing problems has led to a surge of interest in designing efficient algorithms to address them. In this paper, we follow the line of research and study the problem of packing convex polygons without rotations in various settings depending on the type of containers used. 

Past research focusing on theoretical considerations of two-dimensional packings mainly concentrates on the scenario when all objects are axis-parallel rectangles. In this paper, we will discuss packing without rotations, in which only translations are permitted.
There are two main classes of the problem depending on whether the size of the container is fixed and we want to minimize the number of containers used or whether we want to minimize the container's size with respect to some objective function. 

A seminal example of the first class is the \textsc{Geometric Bin Packing} problem in which a number of unit size squared bins to pack is to be minimized. The problem is arguably the most natural generalization of the regular $(1D)$ \textsc{Bin Packing} to two dimensions, and its absolute approximability has been fully understood. Unless $\mathcal{P} = \mathcal{NP}$, the best possible efficient constant factor approximation is $2$ \cite{leung}, and such an algorithm is known \cite{harren_vanstee}.

In the second class, there are several variants to be considered. An example is the \textsc{Strip Packing Problem} which is concerned with packing objects into a strip of width $1$ and infinite height in such a way that the maximum of all the heights of the placed objects is minimized. Like \textsc{Geometric Bin Packing}, \textsc{Strip Packing} generalizes $(1D)$ \textsc{Bin Packing}. The best known efficient constant factor approximation for \textsc{Strip Packing} has approximation factor $(5/3+ \epsilon)$~~\cite{harren}. It is known that there can not exist a polynomial time algorithm with an approximation ratio of $(3/2-\epsilon)$ for any $\epsilon > 0$ unless $\mathcal{P} = \mathcal{NP}$, which follows directly from the approximation hardness of $(1D)$ \textsc{Bin Packing}. Both classes of problems have been also considered in the asymptotic setting, see e.g.~\cite{hoberg, khanbansal,JANSEN_strip}.

In younger time, there was also an increase of interest in cases where the objects in question are general convex polygons.
Alt, de Berg and Knauer \cite{alt_17,alt_17_corrigendum} considered the problem of packing an instance consisting of a number of convex polygons of the form $p \subset [0,1]^2$ into a minimum area axis-parallel rectangular container. We refer to this problem as \textsc{Polygon Area Minimization} throughout this paper. In the special case where the instance consists of rectangles only, the problem is known to admit a PTAS~\cite{bansal_06}. They proved the existence of the following efficient algorithm:
\begin{itemize}
    \item A $23.78$-approximation for \textsc{Polygon Area Minimization}.
\end{itemize}
Recently, Aamand, Abrahamsen, Beretta and Kleist \cite{aamand_23} showed that the algorithm of Alt, de Berg and Knauer can be leveraged to obtain also efficient approximation algorithms of further problems:
\begin{itemize}
    \item A $7$-approximation for \textsc{Polygon Perimeter Minimization}.
    \item A $51$-approximation for \textsc{Polygon Strip Packing}.
    \item An $11$-approximation for \textsc{Polygon Bin Packing} for polygons with diameter at most $\frac{1}{10}$.
\end{itemize}

\subsection{Our results}
The results of Alt, de Berg and Knauer \cite{alt_17,alt_17_corrigendum} and Aamand et al. \cite{aamand_23} are heavily based on so-called shelf packing algorithms. In shelf packing algorithms, the objects are first placed on the shelves, possibly ordered by height, which are later stacked on one another to build a final solution. Compared to axis-parallel rectangles, the main challenge in designing an approximation algorithm for polygon packing problems is that objects cannot be sorted by height and orientation simultaneously. As a result, the algorithm and its analysis in~\cite{alt_17,alt_17_corrigendum} have such a large approximation guarantee. In this paper, we provide new insight into how shelf packing algorithms should be applied to polygon packing problems. We introduce a sequence of transformations of the objects that allow us first to sort them by height and later by orientation to build a solution of a much better approximation guarantee. More precisely, we design polynomial-time algorithms with the following factors:
\begin{itemize}
    \item A $9.45$-approximation for \textsc{Polygon Area Minimization}.
\end{itemize}
Using this algorithm as a subroutine, we build upon the methods from Aamand et al. \cite{aamand_23} to obtain the following efficient approximation algorithms:
\begin{itemize}
    \item A $(3.75+\epsilon)$-approximation for \textsc{Polygon Perimeter Minimization}.
    \item A $21.89$-approximation for \textsc{Polygon Strip Packing}.
    \item A $5.09$-approximation for \textsc{Polygon Bin Packing} for polygons which have their diameter bounded by $\frac{1}{10}$.
\end{itemize}

The results are proved in Sections~\ref{section_polygon_area_min},~\ref{section_polygon_perimeter_min},~\ref{section_polygon_strip}, and~\ref{section_polygon_1/M} respectively.
Furthermore, in Section \ref{section_polygon_bin_pack} we show the following results, which make progress towards solving an open question of a $\mathcal{O}(1)$-approximation algorithm for \textsc{Polygon Bin Packing} for arbitrary polygons.
\begin{itemize}
    \item There is an efficient $\mathcal{O}(\frac{1}{\delta})$-approximation algorithm for \textsc{Polygon Bin Packing} for instances where each polygon has width or height at most $1-\delta$.
    \item There is an efficient $\mathcal{O}(1)$-approximation algorithm for \textsc{Polygon Bin Packing} for instances with the property that all polygons share a spine (up to translation) with height at least $\frac{3}{4}$.
\end{itemize}

\section{Preliminaries} 
We start out considerations by recalling a classical and well-known problem in theoretical computer science, the \textsc{Bin Packing} problem: Given a list of numbers $s_1, \dots, s_n \in (0,1] \cap \Q$, representing the sizes of $n$ objects, the goal is to find the minimum number of bins of size $1$, so that we can pack all objects into them. 
\textsc{Bin Packing} can be seen as the task of packing ($1$-dimensional) intervals into as few intervals of length $1$ as possible. This definition can be extended to the two dimensional case. To do so we introduce several definitions. 

Throughout the paper, for a set subset $A$ of the domain of a function $f$, $f(A)$ is a shorthand notation for $\sum_{a \in A} f(a)$.
A \textbf{packing instance} is defined by a finite set $I$ containing objects to pack, a countable set $\mathcal{R}$ containing bins to pack the objects into and a set $\Phi$ of allowed transformations.
A \textbf{shape} is a compact connected set $s \subset \R^2_{\geq 0}$. 
We denote a rectangular shape  $r \subset \R^2_{\geq 0}$ by a tuple $r=(w,h) \in \R^2_{>0}$, which has width $w$ and height $h$. 
An \textbf{object} $o$ is an element of $I$ and has a shape $s \subset [0,1]^2$.
Furthermore, we define $\rectangles$ and $\polygons$ to be the sets which contain all finite sets $I$ with objects consisting of axis-parallel rectangles and convex polygons, respectively. We will usually denote $\lvert I \rvert$ by $n$. In order to ensure computability, we assume that each object in $\rectangles$ and $\polygons$ is defined by finitely many vertices.

A \textbf{bin} $R \in \mathcal{R}$ is also characterized by having a shape. In this paper, we always assume that the shape of a bin $R$ is an axis-parallel rectangle. That is, it is a rectangle with each of its sides being parallel to one of the primal axes in $\R^2$, and normally, its lower left corner is at the origin $(0,0) \in \R^2$. 
We write $\width(R)$ and $\height(R)$ for the width and the height of a bin $R \in \mathcal{R}$. If $R=[0,1]^2$, we say that $R$ is a \textbf{unit bin}.
In this study, the set of \textbf{allowed transformations} $\Phi$ is the set of all translations, i.e. $\Phi = \{\phi: R^2 \to \R^2 \, | \, \exists x_0 \in \R^2 \, \forall x \in \R^2: \phi(x) = x + x_0 \}$. Note that we consider the setting where rotations or reflections of objects are not allowed.
We denote the width and the height of an object $o$ (the length of the projection on the $x$-axis or $y$-axis respectively) by $\width(o)$ and $\height(o)$. The maximum width and height of a shape of an object in $I$, we denote by $w_{\max}(I)$ and $h_{\max}(I)$ respectively. We will also use the notation $\width(S)$ and $\height(S)$ for arbitrary subsets $S \subset \R^2$

A \textbf{packing} $P$ of $I$ is defined as a set of pairwise disjoint placements of all $o \in I$ with respect to the set of translations $\Phi$. 
Given a bin $R \in \mathcal{R}$, we say that we can \textbf{pack $I$ into $R$} if there is a packing $P$ of $I$ so that $P \subset R$. In this case, we also say that $P$ is a \textbf{packing of $I$ into $R$}. If we can pack the objects $I$ into the bin $R$, we call $R$ a \textbf{bounding box} of $I$. 
If we have an at most countable collection of bins $\mathcal{R}:= \{R_j\}_{j \in J}$, we say that we can \textbf{pack $I$ into $\mathcal{R}$}, if there is a partition $I=\bigsqcup_{j \in J} I_j$ so that we can pack the objects $I_j$ into $R_j$ for all $j \in J$. In such case, if for all $j \in J$, $P_j$ is a packing of $I_j$ into $R_j$, we refer to $\mathcal{P}:=\{P_j\}_j$ as a \textbf{multi-packing}. For any $j \in J$, we say that \textbf{$o$ is in the packing $P_j$}, if $o \in I_j$.
We define the $\area(P)$ to be the area of the smallest axis-parallel rectangle containing $P$.

The definition already demonstrates the more difficult nature of multi-dimensional packings compared to one dimensional packings. Even if the objects to pack are axis-parallel rectangles, it is no longer sufficient to consider in which bin to pack which rectangle, but the exact position in the bin also matters.

In this paper we consider problems consisting in packing convex polygonal shapes into axis-parallel rectangular bins under translational transformations. More precisely we consider the following problems.

\begin{problem}[\textsc{Polygon Packing}]\ \\
\begin{tabularx}{\linewidth}{l X}
\textbf{Input:} &Convex polygons $I \in \polygons$. \\
\textbf{Goals:} & 
\end{tabularx}
\begin{itemize}
\item \textsc{Bin Packing:} Find the minimum number $B \in \N$ so that we can pack $I$ into $B$ unit bins.
\item \textsc{Strip Packing:} Find the minimum height $H \in \Q_{>0}$ so that we can pack $I$ into a bin of width $1$ and height $H$.
\item \textsc{Area Minimization:} Find a bounding box $R \in \R^2_{>0}$ of $I$ so that $f(R)=\width(R)\cdot \height(R)$ is minimal.
\item \textsc{Perimeter Minimization:} Find a bounding box $R \in \R^2_{>0}$ of $I$ so that $f(R)=2(\width(R) + \height(R))$ is minimal.
\item \textsc{Minimum Square:} Find a bounding box $R \in \R^2_{>0}$ so that $f(R) = \max\{\width(R),$ $\height(R)\}$ is minimal.
\end{itemize}
Throughout the paper for the problems under consideration, we denote the optimal value of an instance $I \in \polygons$ by $\opt(I)$.

\end{problem}

\section{Shelf Packing Algorithms}
\label{section:shelf_packing}

Introducing well-known shelf-packing algorithms involves basic $(1D)$-\textsc{Bin Packing} algorithms, such as \textsc{NextFit} (NF), \textsc{FirstFit} (FF), and \textsc{BestFit} (BF). These place items $s_1, \dots, s_n$ into bins sequentially, with $s_{i+1}$ placed according to specific rules. If no placement adheres to the rule, a new bin is opened. The rules differ for each algorithm.

\begin{itemize}
    \item NF, places $s_{i+1}$ into the most recently opened bin, if it has enough space.
    \item FF, places $s_{i+1}$ in the earliest opened bin in which it fits.
    \item BF, places $s_{i+1}$ in the bin with least free space among bins in which $s_{i+1}$ fits.
\end{itemize}


NF and FF are often preprocessed by sorting the items in non-increasing size, which are called \textsc{NextFitDecreasing} (NFD) and \textsc{FirstFitDecreasing} (FFD). It is not hard to show that NF is a $2$-approximation for the \textsc{Bin Packing} problem~\cite{johnson_phd}. Moreover, NF packs it into at most $1 + 2 \sum_{i=1}^n s_i$ bins. BF and FF both have an approximation ratio of $1.7$ \cite{dosa_sgall,dosa_sgall_bf}, which are tight. Furthermore, as shown in \cite{johnson_1974}, if $s_i \leq \frac{1}{m}$ for all $i \in [n]$ and some $m \geq 2$, FF packs this instance into at most $1 + (1 + \frac{1}{m} )\sum_{i=1}^n s_i$ bins. \footnote{In fact they show that FF packs such instance in $2 + (1 + \frac{1}{m} )\sum_{i=1}^n s_i$ bins. With a strategy analogous to the one from the proof of Theorem 3 in \cite{coffman_strip} in the two-dimensional case, however, one can show that an additive factor of $1$ is sufficient.}


Variants of NF, FF and BF exist for 2D rectangle packing problems, called \textsc{NextFitDecreasingHeight} (NFDH), \textsc{FirstFitDecreasingHeight} (FFDH), and \textsc{BestFitDecreasingHeight} (BFDH). These shelf-packing algorithms were introduced for \textsc{Strip Packing}, placing rectangles $r_1, \dots, r_n \in \rectangles$ into a bin $R=[0,1]\times[0,\infty)$ with infinite height.
%
%
The three algorithms order rectangles $r_1, \dots, r_n$ in non-increasing height and place them sequentially into shelves in $R$. A shelf is a horizontal strip, and rectangles can open new shelves. A shelf-packing algorithm places $r_{i+1}$ in an existing shelf according to a rule or opens a new shelf. The placement is bottom-left without intersecting other rectangles. 

\begin{itemize}
    \item NFDH places $r_{i+1}$ in the most recently opened shelf, if it fits.
    \item FFDH places $r_{i+1}$ in the lowest possible shelf which has enough space.
    \item BFDH, as FFDH, allows for placing rectangles in a lower shelf than the most recently opened one. Here, $r_{i+1}$ is placed in the one which has minimal free horizontal space at the right end of the shelf among all shelves, while still having at least $\width(r_{i+1})$ of it.
\end{itemize}

%
Since NFDH and FFDH are of particular interest in our paper, we present the following two absolute approximability results for axis-parallel rectangle \textsc{Strip Packing} problem.

\begin{theorem}[Theorem 1~\cite{coffman_strip}]
Let $I \in \rectangles$. The packing $P$ obtained by \normalfont{NFDH} satisfies
\[
\height(P) \leq h_{\max}(I) + 2 \area(I) \leq 3 \opt(I).
\]
\label{nfdh}
\end{theorem}

\begin{theorem}[Theorem 3~\cite{coffman_strip}]
Let $m \in \N$ and $I \in \rectangles$ be satisfying that $w_{\max}(I) \leq \frac{1}{m}$. The packing $P$ of $I$ obtained by \normalfont{FFDH} satisfies
\[
\height(P) \leq h_{\max}(I) + \bigg(1+\frac{1}{m}\bigg) \area(I) \leq \bigg( 2 + \frac{1}{m}\bigg) \opt(I).
\]
\label{thm:FFDH}
\end{theorem}

\section{An Efficient $9.45$-Approximation for \textsc{Polygon Area Minimization} and $7$-Approximation when all Polygons are $x$-Parallelograms}
\label{section_polygon_area_min}

In this section, we prove that there is an efficient $9.\overline{4}$-approximation algorithm for \textsc{Polygon Area Minimization}, which improves the previous best approximation factor for polynomial time algorithms of $23.78$ by Alt et al. \cite{alt_17,alt_17_corrigendum}. Based on that result, we also show a $7$-approximation in the special case when all polygons are $x$-parallelograms, a special type of parallelograms we will introduce in Definition \ref{def_x_para}.

\subsection{An Efficient 9.45-Approximation for \textsc{Polygon Area Minimization}}
To start the discussions, we introduce the following definition.

\begin{definition}
\label{def_spine}
Let $p \subset [0,1]^2$ be a polygon. A \textbf{spine} $s$ of $p$ is a (straight) line segment connecting a point in $\argmin_{(x,y) \in p} y$ with a point in $\argmax_{(x,y) \in p} y$. We call the angle between the $x$-axis in increasing direction and $s$ the \textbf{angle of $s$}. We say that $s$ \textbf{is tilted to the right} or \textbf{leans to the right}, if this angle is in $(0,\frac{\pi}{2}$] and \textbf{is tilted to the left} or \textbf{leans to the left}, if it is in $[\frac{\pi}{2},\pi)$.
\end{definition}

Sometimes we also talk about ``the'' spine of a polygon, implicitly assuming that one has been fixed. The algorithm of Alt et al. is a shelf-packing algorithm and ordering polygons by the angle of their spines is a crucial step. Our algorithm shares these two characteristics. A key difference is that our algorithm first packs the polygons into parallelograms that have two of their sides parallel to the $x$-axis. We give such parallelograms their own definition:

\begin{definition}
\label{def_x_para}
An \textbf{$x$-parallelogram} is a parallelogram $q \subset \R^2$ that has two of its sides parallel to the $x$-axis. Of those two sides, we call the one with lower $y$-coordinate the \textbf{base} of $q$. We write $\base(q)$ for the length of the base of $q$ and $\wside(q)$ for the width of one of the sides of $q$ which is not parallel to the $x$-axis. Similar as to the definition for spines of polygons, we say that $q$ \textbf{is tilted to the right} or \textbf{leans to the right}, if the angle between its right side and the increasing direction of the $x$-axis is in $(0,\frac{\pi}{2}$] and \textbf{is tilted to the left} or \textbf{leans to the left}, if it is in $[\frac{\pi}{2},\pi)$. We refer to this angle simply as \textbf{angle of $q$}.
\end{definition}

For packing polygons into $x$-parallelograms, we prove the following result, which is a refined version of the discussions of Aamand et al. \cite{aamand_23} in Subsection $5.2.4$ of their paper.

\begin{lemma}
\label{polygon_into_parallelogram}
Let $p$ be a convex polygon. Then there exists an $x$-parallelogram $q$ that contains $p$ and satisfies
\begin{enumerate}[(i)]
    \item $\height(q) = \height(p)$
    \item $\base(q) \leq \width(p)$
    \item $\wside(q) \leq \width(p)$
    \item $\area(q) \leq 2 \area(p)$.
\end{enumerate}
\end{lemma}

\begin{proof}
First, we construct a bounding $x$-parallelogram as the one that can be seen in Figure \ref{fig_poly_in_para}. Let $l_b$ and $l_t$ be lines parallel to the $x$-axis and tangent to $p$, touching the bottom of $p$ and the top of $p$ respectively. Choose $p_b \in p \cap l_b$ and $p_t \in p \cap l_t$ and define $s$ to be the line connecting $p_b$ and $p_t$. Note that $s$ is a spine of $p$. Let $s_l$ and $s_r$ be tangent to $p$ and parallel to $s$, lying on the left and on the right of $p$. Let $p_l \in p \cap s_l$ and $p_r \in p \cap s_r$. We now define $q$ to be the set bounded by $l_b,l_t,s_l$ and $s_r$. Note that $q$ is an $x$-parallelogram that satisfies $(i)$. As the left and right sides of $q$ are just translations of $s$, it also satisfies $(iii)$, because of course $s \subset p$ by convexity of $p$.

\begin{figure}[h]
\centering
  \includegraphics[width=.5\linewidth]{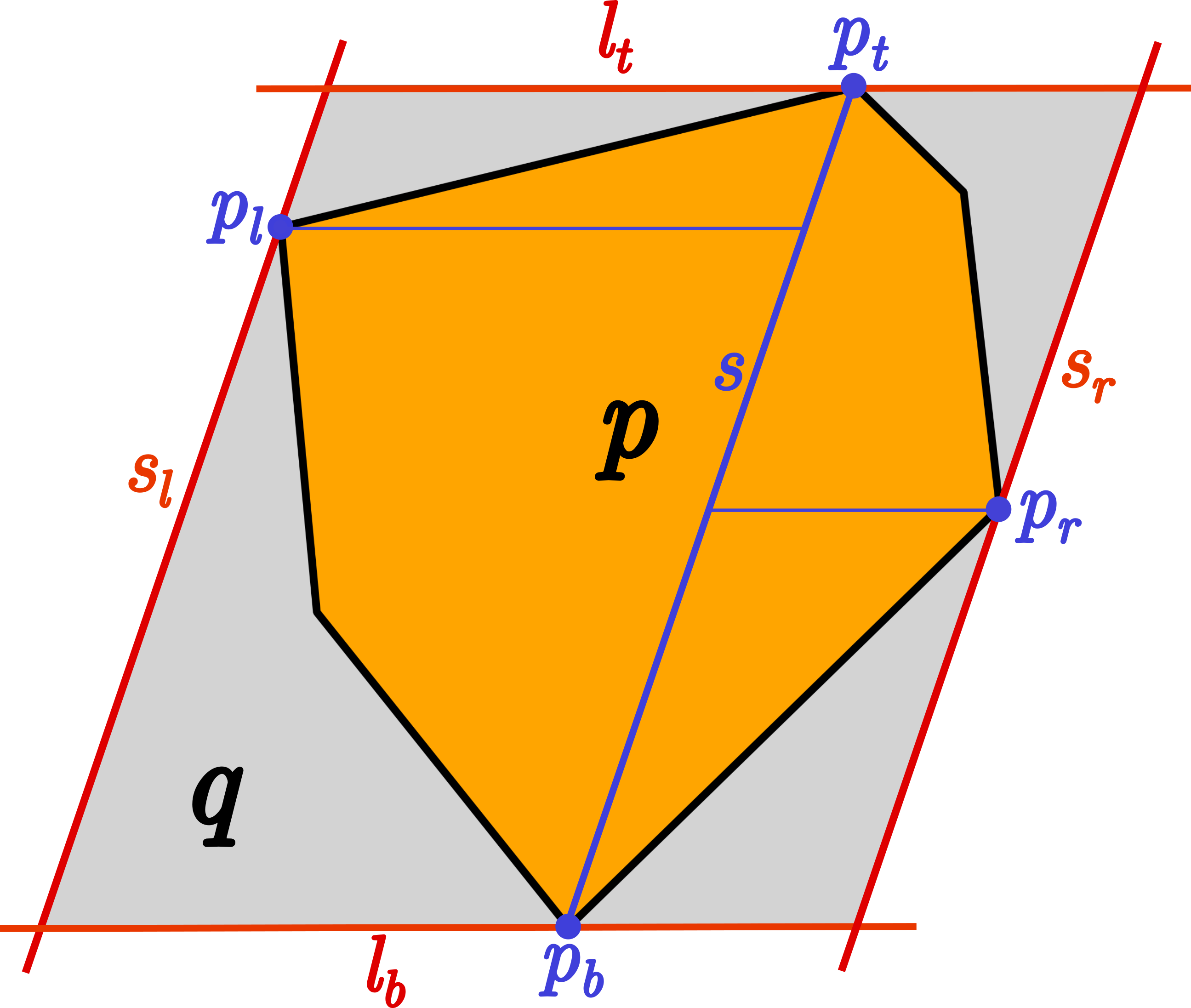}
  \caption{The construction of a bounding parallelogram $q$ from the proof of Lemma \ref{polygon_into_parallelogram} on an example polygon $p$. Note that here, it holds that $\base(q) > \width(p)$.}
  \label{fig_poly_in_para}
\end{figure}

To see that $q$ also satisfies $(iv)$, we consider the triangle with vertices $p_t,p_b$ and $p_l$. We note that it is contained in $p$ due to convexity and contains exactly half the area of the part of $q$ that lies on the left of $s$. Analogously, the triangle with vertices $p_t,p_b$ and $p_r$ has half the area of the part of $q$ that lies on the right of $s$. So $q$ does indeed also satisfy $(iv)$.

However, $q$ does not necessarily satisfy $(ii)$ (see the polygon in Figure \ref{fig_poly_in_para}). If it does, then $q$ satisfies all assumptions of the lemma and we are done.

So we consider now the case when $\base(q) > \width(p)$. Consider the axis-parallel rectangle $r$ bounding $p$. That is, $r$ contains $p$ and each of its sides has non-empty intersection with $p$. Surely $r$ satisfies $(i), (ii)$ and $(iii)$. To see that it also satisfies $(iv)$, note that
\begin{align*}
\area(r) &= \width(r) \height(r) \\
&= \width(p) \height(p) \\
& < \base(q) \height(q) \\
& = \area(q) \\
& \leq 2 \area(p).
\end{align*}
So whenever $\base(q) > \base(p)$, we showed that $r$ satisfies all requirements $(i)-(iv)$ of the lemma instead. Since $r$ is also an $x$-parallelogram, we conclude.
\end{proof}

With the help of Lemma \ref{polygon_into_parallelogram}, we can now present the main result of this chapter.

\begin{theorem}
\label{polygon_packing_approx}
There is a polynomial time $9.\overline{4}$-approximation algorithm for \textsc{Polygon Area Minimization}. Moreover, there is such algorithm with running time $\mathcal{O}(n^2+N)$, where $n$ is the number of polygons and $N$ the total number of vertices in a given input $I \in \polygons$, assuming that each polygon $p \in I$ is given as a list of its vertices.
\end{theorem}

\begin{proof}
Let $I \in \polygons$. We construct a packing $P$ as the one depicted in Figure \ref{fig_polygon_packing}.

Pack each polygon $p \in I$ into an $x$-parallelogram $q$ as in Lemma \ref{polygon_into_parallelogram}. Call the instance of all so-obtained $x$-parallelograms $I_Q$. The idea of our algorithm is to use FFDH to pack straightened, axis-parallel versions of the $x$-parallelograms $I_Q$ and to then use this packing to obtain one for $I_Q$ and hence also $I$ which is not much bigger.

So, define yet another instance $I_R$ which, for each $q \in I_Q$, contains a rectangle $r=(\base(q),\height(q))$. With FFDH, we now pack $I_R$ into a strip of width $c w_{\max}(I)$, where $c \geq 1$ is to be determined later. Call the so-obtained packing $P_R$. By Theorem \ref{thm:FFDH} it follows
\begin{equation}
\label{area_ffdh_polygon_pack}
\height(P_R)(cw_{\max}(I)) \leq \bigg (1+\frac{1}{m} \bigg)\area(I_R) + c h_{\max}(I_R) w_{\max}(I),  
\end{equation}
where $m = \lfloor c \rfloor$, as $w_{\max}(I_R) = \max_{q \in I_Q} \base(q) \leq w_{\max}(I)$ by Lemma \ref{polygon_into_parallelogram}.

Let $S \subset I_Q$ be the parallelograms corresponding to the rectangles in a certain shelf in the packing $P_R$. We can pack $S$ into a new shelf of width $(c+2)w_{\max}(I)$ by first ordering the parallelograms $S$ by decreasing angle. Indeed, if we, after this ordering, place them all next to each other in the shelf, we note that now all bases of the parallelograms are connected to each other and hence the overlap on either side is at most $\max_{q \in S} \wside(q) \leq w_{\max}(I)$ by Lemma \ref{polygon_into_parallelogram}. Put all such shelves on top of each other and call the so-obtained packing $P_Q$. Note that $P_Q$ has the same height as $P_R$, but $\frac{c+2}{c}$ times its width. Finally, we pack each polygon into its respective parallelogram in the packing $P_Q$. Call this packing $P$. Then
\begin{align*}
\area(P) &\leq \area(P_Q) \\
&\leq \frac{c+2}{c}\height(P_R)(cw_{\max}(I)) \\
& \leq \frac{c+2}{c}\bigg(1+\frac{1}{m}\bigg)\area(I_R) + (c+2) h_{\max}(I_R) w_{\max}(I) \\ 
&= \frac{c+2}{c}\frac{m+1}{m}\area(I_Q) + (c+2) h_{\max}(I_Q) w_{\max}(I) \\ 
& \leq 2 \frac{c+2}{c} \frac{m+1}{m} \area(I) + (c+2) h_{\max}(I) w_{\max}(I) \\
& \leq \bigg ( 2\frac{c+2}{c}\frac{m+1}{m}+(c+2) \bigg) \opt(I).
\end{align*}
One can check that the minimum is attained at $c=3$, in which case we get a $9.\overline{4}$-approximation.

\begin{figure}[h]
\centering
  \includegraphics[width=.7\linewidth]{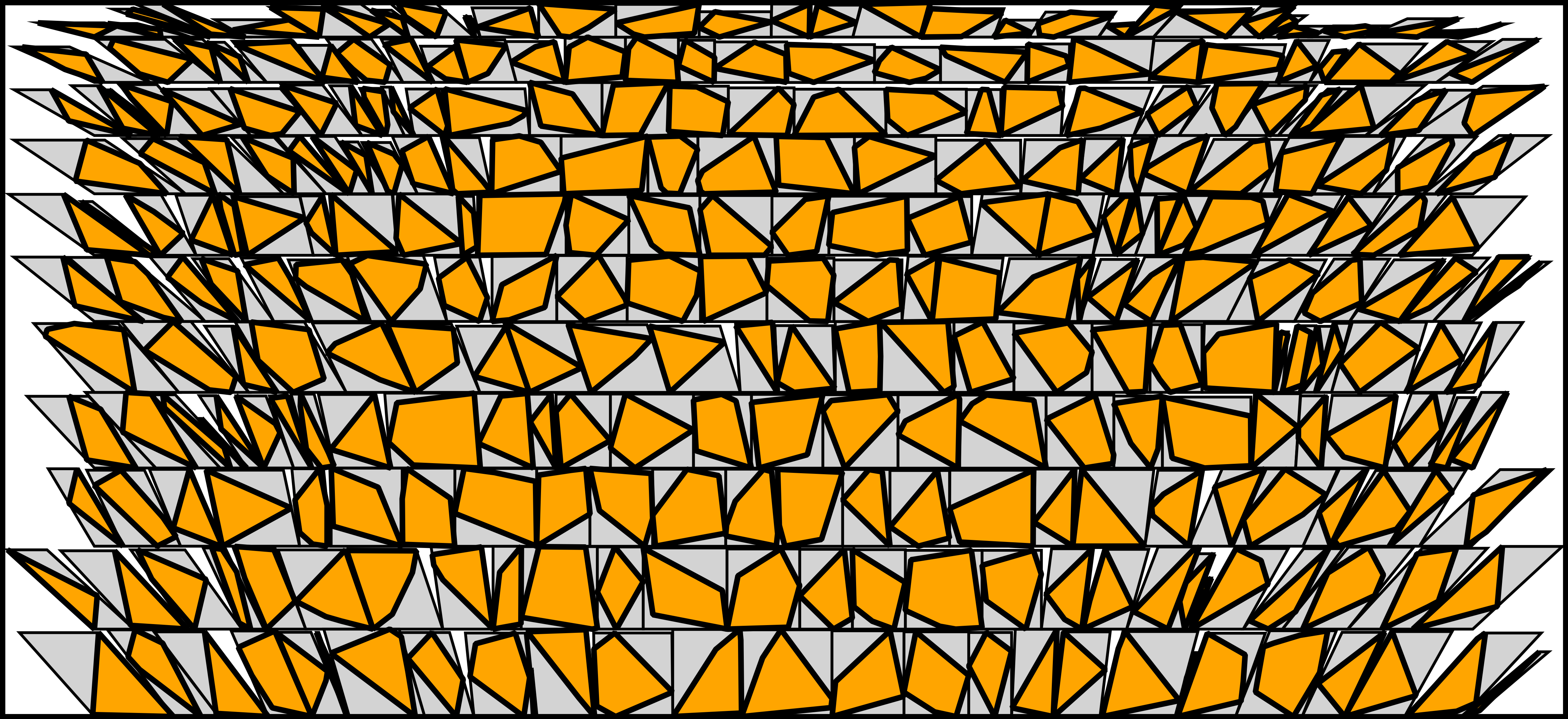}
  \caption{A packing computed with the algorithm from the proof of Theorem \ref{polygon_packing_approx}, here with parameter $c=15$. For each polygon, also its computed bounding $x$-parallelogram is drawn.}
  \label{fig_polygon_packing}
\end{figure}

We now note that the claimed running times of the above algorithm follows by observing that: Constructing $I_Q$ from $I$ can be done in $\mathcal{O}(N)$ time, constructing $I_R$ from $I_Q$ can be done in $\mathcal{O}(n)$ time, constructing $P_R$ can be done in $\mathcal{O}(n^2)$ time, constructing $P_Q$ from $P_R$ can be done by sorting in $\mathcal{O}(n \log(n))$ time, and finally, constructing $P$ from $P_Q$ can be done in $\mathcal{O}(N)$ time.
\end{proof}

The algorithm of Alt et al. runs in time $\mathcal{O}(N \log(N))$ and thus for certain instances faster than our algorithm. If, in our algorithm, the use of FFDH for packing $I_R$ is replaced by NFDH, one can show using Theorem \ref{nfdh} that for an optimal choice of $c=2 \sqrt{2}$, the algorithm has an approximation guarantee of 11.66 while having a running time of $\mathcal{O}(n \log(n) + N)$, thus obtaining an algorithm that runs faster than the algorithm of Alt et al., while still having a greatly improved approximation ratio.

\subsection{An Efficient $7$-Approximation for \textsc{Polygon Area Minimization} when all Polygons are $x$-Parallelograms}
\label{section_x_parallelogram_areamin}

In the algorithm presented in the previous subsection, we first pack general polygons into $x$-parallelograms and afterwards pack these $x$-parallelograms. One would expect that if one wants to pack $x$-parallelograms from the start, one should be able to obtain a better approximation factor. Showing that this is indeed the case is the content of this brief section.

\begin{theorem}
There is a polynomial time $7$-approximation algorithm for \textsc{Polygon Area Minimization}, when all input polygons are $x$-parallelograms.
\end{theorem}

The proof follows along the lines of the proof of Theorem \ref{polygon_packing_approx}.

\begin{proof}
Let $I \in \polygons$ be so that every $p \in I$ is an $x$-parallelogram. As in the proof of Theorem \ref{polygon_packing_approx}, we define the set $I_R \in \rectangles$ that for every $p \in I$ contains some $r \in I_R$ with $\width(r) = \base(p)$ and $\height(r) = \height(p)$. Again, we pack $I_R$ with FFDH into a strip of width $cw_{\max}(I)$ and call the so-obtained packing $P_R$.

After ordering them by angle, we can pack the parallelograms $S \subset I$ corresponding to the rectangles in some shelf in $P_R$ into a new shelf of width $(c+2)w_{\max}(I)$, because of course $\wside(p) \leq w_{\max}(I)$. Therefore, calling the so-obtained packing $P$,
\begin{align*}
    \area(P) &\leq \frac{c+2}{c}\area(P_R) \\
    & \leq \frac{c+2}{c}\frac{m+1}{m}\area(I) + (c+2)h_{\max}(I)w_{\max}(I) \\
    & \leq \bigg( \frac{c+2}{c}\frac{m+1}{m} + (c+2) \bigg) \opt(I),
\end{align*}
which, for $c=2$, is minimized and equal to $7$.
\end{proof}

\section{An Efficient $(3.75+\epsilon)$-Approximation for \textsc{Polygon Perimeter Minimization} and $(3.56+\epsilon)$-Approximation for \textsc{Polygon Minimum Square}}
\label{section_polygon_perimeter_min}

In this section, we show how to leverage the algorithm from Theorem \ref{polygon_packing_approx} to obtain an approximation-algorithm for \textsc{Polygon Perimeter Minimization}. We improve on the polynomial time $7.3$-approximation algorithm obtained by Aamand et al. \cite{aamand_23} and present an efficient $(3.75 + \epsilon)$-approximation. Moreover, we show how to leverage that result to obtain an $(3.56+\epsilon)$-approximation-algorithm for \textsc{Polygon Minimum Square}.

\subsection{An Efficient $(3.75+\epsilon)$-Approximation for \textsc{Polygon Perimeter Minimization}}

\begin{theorem}
\label{polygon_perimeter_min}
For every $\epsilon>0$, there is an efficient $(3.75+\epsilon)$-approximation algorithm for \textsc{Polygon Perimeter Minimization}.
\end{theorem}

\begin{proof}
Let $I \in \polygons$. Note that for the perimeter objective, it surely holds that
\begin{equation}
\label{opt_perimeter_vs_wmax_hmax}
\opt(I) \geq 2(w_{\max}(I) + h_{\max} (I)).
\end{equation}
Furthermore, since a bounding box of $I$ has area at least $\area(I)$ and the minimum perimeter rectangle having area $\area(I)$ is a square, it also is true that
\[
\opt(I) \geq 4 \sqrt{\area(I)}.
\]
It follows from Inequality \ref{opt_perimeter_vs_wmax_hmax} that
\[
\min \{h_{\max}(I),w_{\max}(I)\} \leq \frac{1}{4}\opt(I)
\]
and without loss of generality, we assume that $w_{\max}(I) \leq \frac{1}{4}\opt(I)$. Otherwise we are making the following argument by packing into vertical shelves instead.

Let $P$ be the packing obtained from the algorithm in Theorem \ref{polygon_packing_approx}, leaving $c$ as a free parameter. Then $P$ satisfies
\[
\width(P) \leq (c+2)w_{\max}(I)
\]
and by dividing the inequality (\ref{area_ffdh_polygon_pack}) by the width of the strip used for the rectangle packing obtained by FFDH, $cw_{\max}(I)$, we see that
\[
\height(P) \leq 2 \frac{m+1}{m} \frac{\area(I)}{cw_{\max}(I)} + h_{\max}(I) \leq \frac{1}{8} \frac{m+1}{m} \frac{\opt(I)^2}{cw_{\max}(I)} + h_{\max}(I).
\]
We restrict the domain of $c$ so that $c \geq \frac{\opt(I)}{4w_{\max}(I)} \geq 1$ and write $c = l \frac{\opt(I)}{4w_{\max}(I)}$ for some $l \geq 1$. Then
\[
\width(P) \leq \frac{l}{4} \opt(I) + 2w_{\max}(I), \quad \height(P) \leq \frac{1}{2l}\frac{m+1}{m} \opt(I) + h_{\max}(I).
\]
The perimeter of $P$ is hence bounded by
\begin{align*}
2(\width(P) + \height(P)) &\leq 2 \bigg( \bigg(\frac{l}{4} + \frac{m+1}{2lm} \bigg) \opt(I) + 2w_{\max}(I) + h_{\max}(I) \bigg) \\
& \leq 2 \bigg ( \frac{1}{4} \bigg(l + \frac{2(m+1)}{lm} \bigg) + 1 \bigg) \opt(I) \\
& = \frac{1}{2} \bigg( l+\frac{2(m+1)}{lm} + 4 \bigg) \opt(I),
\end{align*}
which is minimized and equal to $3.75 \opt(I)$ when $l$ is equal to $\overline{l}:=2$, where we used that $m = \lfloor c \rfloor \geq \lfloor l \rfloor$.

Now since the value of $\opt(I)$ is not known beforehand and since $\overline{l}$ depends on it, we need to guess an optimal value for $l$. This can be done as follows. Compute the packing from the algorithm in Theorem \ref{polygon_packing_approx} for $c=1,(1+\epsilon), \dots, (1+\epsilon)^K$, where $K=\frac{\log(n)}{\log(1+\epsilon)}$ and denote the packing obtained for $c=(1+\epsilon)^k$ by $P_k$ for all $k \in \{0,1, \dots, n\}$. Over all those packings, choose the one that has minimum perimeter. Say this perimeter is $z > 0$. Let $k \in \{1, \dots, K\}$ be so that
\[
(1+\epsilon)^{k-1} \leq \overline{c} \leq (1+\epsilon)^{k},
\]
where $\overline{c} := \overline{l} \frac{\opt(I)}{4w_{\max}(I)}$. Then, for $l_k := (1+\epsilon)^k \frac{4w_{\max}(I)}{\opt(I)}$, it holds that
\[
l_k = (1+\epsilon)^k \frac{4w_{\max}(I)}{\opt(I)} \leq (1+\epsilon) \overline{c}\frac{4w_{\max}(I)}{\opt(I)} = (1+\epsilon)\overline{l}.
\]
In particular, as of course also $\overline{l} \leq l_k$, it holds that
\begin{align*}
z & \leq 2(\width(P_k) + \height(P_k)) \\
&\leq \frac{1}{2} \bigg( l_k+\frac{2(m+1)}{l_km} + 4 \bigg) \opt(I) \\
&\leq (1+\epsilon) \frac{1}{2} \bigg( \overline{l}+\frac{2(m+1)}{\overline{l}m} + 4 \bigg) \opt(I) \\
&\leq (1+\epsilon)3.75 \opt(I),
\end{align*}
which shows the statement.
\end{proof}

\subsection{An Efficient $(3.56+\epsilon)$-Approximation for \textsc{Polygon Minimum Square}}
\label{section_polygon_min_square}

In this section, we show how to leverage the algorithm from Theorem \ref{polygon_packing_approx} to obtain an approximation-algorithm for \textsc{Polygon Minimum Square}.

\begin{theorem}
For every $\epsilon>0$, there is an efficient $(3.56+\epsilon)$-approximation algorithm for \textsc{Polygon Minimum Square}.
\end{theorem}

The proof is similar to the proof of Theorem \ref{polygon_perimeter_min}.

\begin{proof}
Let $I \in \polygons$. Note that
\[
\opt(I) \geq \max\{w_{\max}(I),h_{\max}(I)\}
\]
and also
\[
\opt(I) \geq \sqrt{\area(I)}.
\]
Analogously to the proof of Theorem \ref{polygon_perimeter_min}, but substituting $c=l\frac{\opt(I)}{w_{\max}(I)}$ for some $l \geq 1$ instead, we can construct a packing $P$ with
\[
\width(P) \leq l \opt(I) + 2w_{\max}(I), \quad \height(P) \leq \frac{2}{l}\frac{m+1}{m}\opt(I) + h_{\max}(I).
\]
Then
\[
\max\{\width(P),\height(P)\} \leq \max \bigg\{l+2,\frac{2(m+1)}{lm}+1 \bigg\} \opt(I).
\]
The minimum is attained for $l=\overline{l}:=\frac{1}{2}(\sqrt{17}-1)$ in which case $m=1$ and the approximation factor is equal to to $\frac{1}{2}(\sqrt{17}+3) \approx 3.56$. As in the proof of Theorem \ref{polygon_perimeter_min}, we can guess the value of $\overline{l}$ to obtain a $(3.56+\epsilon)$-approximation algorithm.
\end{proof}



\section{An Efficient $21.89$-Approximation for \textsc{Polygon Strip Packing}}
\label{section_polygon_strip}

In this section, we show how to obtain a polynomial time $21.\overline{8}$-approximation for \textsc{Polygon Strip Packing}, improving on the previous best known approximation factor for efficient algorithms of $51$ by Aamand et al. \cite{aamand_23}.
The idea is to construct vertical shelves with the algorithm from Theorem \ref{polygon_packing_approx} and then to stack such vertical shelves horizontally into the strip. This idea comes from \cite{aamand_23}. We slightly improve their procedure using the following observation.

\begin{lemma}
\label{split_shelf}
Let $\overline{I} \in \polygons$ and let $\overline{P} \subset [0,(c+2)w_{\max}(\overline{I})] \times [0, \infty)$ be a shelf-packing obtained by the algorithm from Theorem \ref{polygon_packing_approx} for some $c \geq 1$. Let $I \subset \overline{I}$ be the polygons in one of the shelves of $\overline{P}$ and define the packing $P:=\{\overline{P}(p)\}_{p \in I}$. Then there is a shelf-packing $P'$ of $I$ into two shelves with $\width(P') \leq \frac{1}{2}(c+3)w_{\max}(\overline{I})$ and $\height(P') \leq 2\height(P)$.
\end{lemma}

\begin{proof}
The idea is to simply split the shelf in half in its middle as follows. Let $x_{\text{mid}} := \frac{(c+2)w_{\max}(\overline{I})}{2}$ and partition $I$ into the two sets
\[
I_L:=\bigg \{p \in I \, | \, \width \big(P(p) \cap (\R_{\geq x_{\text{mid}}} \times \R) \big) \leq \frac{\width(p)}{2} \bigg \}
\]
and $I_R = I \backslash I_L$. Note that
\[
I_R \subset \bigg \{p \in I \, | \, \width \big(P(p) \cap (\R_{\leq x_{\text{mid}}} \times \R) \big) < \frac{\width(p)}{2}\bigg\}.
\]
We can now pack $I_L$ and $I_R$ into separate shelves while respecting the ordering of the polygons in $P$. We denote the packing where both of those shelves are stacked on each other by $P'$. Both shelves have width at most
\[
x_{\text{mid}} + \frac{w_{\max}(I)}{2} \leq x_{\text{mid}} + \frac{w_{\max}(\overline{I})}{2} = \frac{1}{2}(c+3)w_{\max}(\overline{I}).
\]
and hence $P'$ does as well.
\end{proof}

Making use of the lemma, we can show the following result.

\begin{theorem}
There is a polynomial time $21.\overline{8}$-approximation algorithm for \textsc{Polygon Strip Packing}.
\end{theorem}

\begin{proof}
Let $I \in \polygons$. We apply the algorithm from Theorem \ref{polygon_packing_approx}, for some $c$ which we will fix later, to construct vertical shelves. More precisely, we rotate each item by an angle of $\pi/2$, apply the algorithm to the rotated instance, and then rotate the whole packing back by $\pi/2$. Let $P$ be the so obtained packing and let $S_1, \dots, S_k$ be the vertical shelves of this packing.

We now stack $S_1, \dots, S_k$ horizontally into the (vertical) strip. That is, we pack the $\{S_i\}_{i \in [k]}$ into horizontal shelves $T_1, \dots, T_l$ themselves. Note that since each shelf $S_i$, where $i \in [k]$, has the same height, the problem reduces to $(1D)$-\textsc{Bin Packing}. In particular, stacking the $\{S_i\}_{i \in [k]}$ greedily, each except for possibly the last horizontal shelf is covered by at least half by the $S_i$. If the last shelf $T_l$ is covered by half, we need at most $l \leq \frac{2\area(P)}{(c+2)h_{\max}(I)}$ horizontal shelves to pack $S_1, \dots, S_k$. Otherwise, we need at most $l \leq \frac{2\area(P)}{(c+2)h_{\max}(I)}+1$ shelves. However, in this case we can reduce the height of $T_l$, which is filled less than half by the $\{S_i\}_{i \in [k]}$, to $\frac{1}{2}(c+3)h_{\max}(I)$ by using Lemma \ref{split_shelf} on every shelf $S$ in $T_l$.

In particular, for such packing $P'$ it holds that
\begin{align*}
\height(P') &\leq \bigg( \frac{2\area(P)}{(c+2)h_{\max}(I)} \bigg) (c+2)h_{\max}(I) + \frac{1}{2}(c+3) h_{\max}(I) \\
&= 2\area(P) + \frac{1}{2}(c+3)h_{\max}(I) \\
&\leq \bigg ( 4\frac{c+2}{c}\frac{m+1}{m}+\frac{5}{2}c + \frac{11}{2} \bigg) \opt(I) \\
& = 21.\overline{8} \opt(I),
\end{align*}
which is attained for $c = 3$.
\end{proof}


\section{Efficient Approximation Algorithms for \textsc{Polygon Bin Packing} for Polygons of Width Upper Bounded by $1/M$}
\label{section_polygon_1/M}

In this section, we present polynomial time approximation algorithms for \textsc{Polygon Bin Packing} for the case when polygons have their width bounded by a fraction of the form $\frac{1}{M}$ and also improved approximations if their height is bounded as well.

To the best of our knowledge, the only presently published polynomial time approximation algorithm for \textsc{Polygon Bin Packing} is an $11$-approximation in the case when the diameter of the input polygons is bounded by $\frac{1}{10}$, see \cite{aamand_23}. For this case, we obtain an efficient $5.09$-approximation.

We prove the following theorem.

\begin{theorem}
\label{polygon_bin_packing_1/M}
Let $I \in \polygons$ and assume that there is some $M \in \N_{\geq 2}$ with $w_{\max}(I) \leq \frac{1}{M}$. Then one can pack $I$ into
\[
\begin{cases}
32 \opt(I) + 5 & \text{if $M=2$} \\
\frac{4M(M-1)}{(M-2)^2}\opt(I) + 3 & \text{if $M \geq 3$}
\end{cases}
\]
bins efficiently. If additionally $h_{\max}(I) \leq \frac{1}{M}$, then we can even efficiently pack $I$ into
\[
\begin{cases}
24 \opt(I) + 3 & \text{if $M=2$} \\
\frac{2(M+1)(M-1)}{(M-2)^2}\opt(I)+2 & \text{if $M \geq 3$}
\end{cases}
\]
bins.
\end{theorem}

In particular, for inputs where both width and height of all polygons are upper bounded by $\frac{1}{10}$, we get a $5.09$-approximation.

\begin{proof}
Assume first that $M \geq 3$. We use the algorithm from Theorem \ref{polygon_packing_approx} for $c=\frac{1}{w_{\max}(I)}-2 \geq M-2$ to obtain a packing $P$ with shelves $S_1, \dots, S_k$. The packing $P$ satisfies
\[
\width(P) \leq (c+2)w_{\max}(I) = 1,
\]
and hence a shelf $S_i$, $i \in [k]$, fits into a unit bin. Note that $m = \lfloor c \rfloor \geq M-2$, since $M-2$ is an integer. Since $cw_{\max}(I) = 1-2w_{\max}(I) \geq \frac{M-2}{M}$, it holds that
\[
\height(P) \leq 2\bigg(1+\frac{1}{m}\bigg)\frac{\area(I)}{cw_{\max}(I)} + h_{\max}(I) \leq 2\bigg(1+\frac{1}{M-2}\bigg)\frac{\area(I)}{(M-2)/M} + h_{\max}(I).
\]
We can now use NF to distribute the shelves into unit bins. Since $h_{max} = 1$, this uses at most $2\height(P)+1$ bins, see Section \ref{section:shelf_packing}. But
\[
2 \height(P) + 1 \leq 4\bigg(1+\frac{1}{M-2}\bigg) \frac{\area(I)}{(M-2)/M} + 2h_{\max}(I) + 1 \leq \frac{4M(M-1)}{(M-2)^2}\opt(I) + 3.
\]

If $h_{\max}(I) \leq \frac{1}{M}$ as well, we can use FF instead of NF to distribute the shelves into bins, which needs at most $(1+\frac{1}{M})\height(P)+1$ bins, see again Section \ref{section:shelf_packing}. This way, the number of needed bins is bounded by
\[
\bigg(1+\frac{1}{M}\bigg) \height(P) + 1 \leq \frac{2(M+1)(M-1)}{(M-2)^2}\opt(I)+2,
\]
as claimed.

Now we consider the case when $M=2$. Partition $I$ into two sets $I_L$ and $I_R$, where a polygon $p \in I$ belongs to $I_L$ if its chosen $x$-parallelogram in the proof of Theorem \ref{polygon_packing_approx} is tilted to the left and to $I_R$ if it is tilted to the right. We now proceed with the algorithm in the proof of Theorem \ref{polygon_packing_approx} for $c=\frac{1}{w_{max}(I)} - 1 \geq 1$, but for $I_L$ and $I_R$ separately. Call the obtained packings $P_L$ and $P_R$, respectively. Note that since all are tilted to one side only,
\[
\width(P_L), \, \width(P_R) \leq (c+1)w_{\max}(I) = 1
\]
and, using $m=\lvert c \rvert \geq 1$ and $cw_{\max}(I) = 1-w_{\max}(I) \geq \frac{1}{2}$,
\[
\height(P_L) \leq 2\bigg(1+\frac{1}{m}\bigg)\frac{\area(I_L)}{cw_{\max}(I_L)} + h_{\max}(I_L) \leq 8\area(I_L) + h_{\max}(I),
\]
as well as analogously $\height(P_R) \leq 8 \area(I_R) + h_{\max}(I)$.

Hence, for the packing $P$, where $P_R$ is stacked on top of $P_L$, it holds that
\[
\width(P) \leq 1 \quad \text{and} \quad \height(P) \leq 16 \area(I) + 2 h_{\max}(I).
\]
From here, with the same arguments as for $M \geq 3$, we obtain the desired bounds.
\end{proof}

\section{Efficient Approximations for \textsc{Polygon Bin Packing} for Polygons with bounded Width or Height and for Instances, where each Polygon has a Spine from a Set of Constant Size}
\label{section_polygon_bin_pack}

In this last section, we give further results on \textsc{Polygon Bin Packing}, attempting to further close the gap towards the question of an efficient $\mathcal{O}(1)$-approximation.

\subsection{An Efficient $\mathcal{O}(\frac{1}{\delta})$-Approximation for \textsc{Polygon Bin Packing} for Polygons with Width or Height at Most $1-\delta$}
\label{section_polygon_1-delta}

We now also consider wider polygons than in Section~\ref{section_polygon_1/M} and prove the following result.

\begin{theorem}
\label{polygon_bin_packing_1_minus_delta}
There is an efficient $\mathcal{O}(\frac{1}{\delta})$-approximation algorithm for \textsc{Polygon Bin Packing} for instances $I \in \polygons$, where each $p \in I$ satisfies $\width(p) \leq 1-\delta$ or $\height(p) \leq 1- \delta$.
\end{theorem}

The result leaves the question whether \textsc{Polygon Bin Packing} admits a $\mathcal{O}(1)$-approximation algorithm open. To obtain an affirmative proof, one is left to provide, for any fixed $\delta>0$, an $\mathcal{O}(1)$-approximation algorithm for polygons $p$ that satisfy both $\width(p) \geq 1-\delta$ and $\height(p) \geq 1-\delta$. We will use and see in the proof, that additionally, one can easily pack polygons with area bigger than some fixed threshold $\eta > 0$. Thus, only polygons that are wide, high and have small area (see Figure \ref{image_problematic_polygon}) are problematic for our described methods. We will present some ideas to deal with such polygons in the coming, last section.

\begin{figure}[h]
\centering
  \includegraphics[width=.5\linewidth]{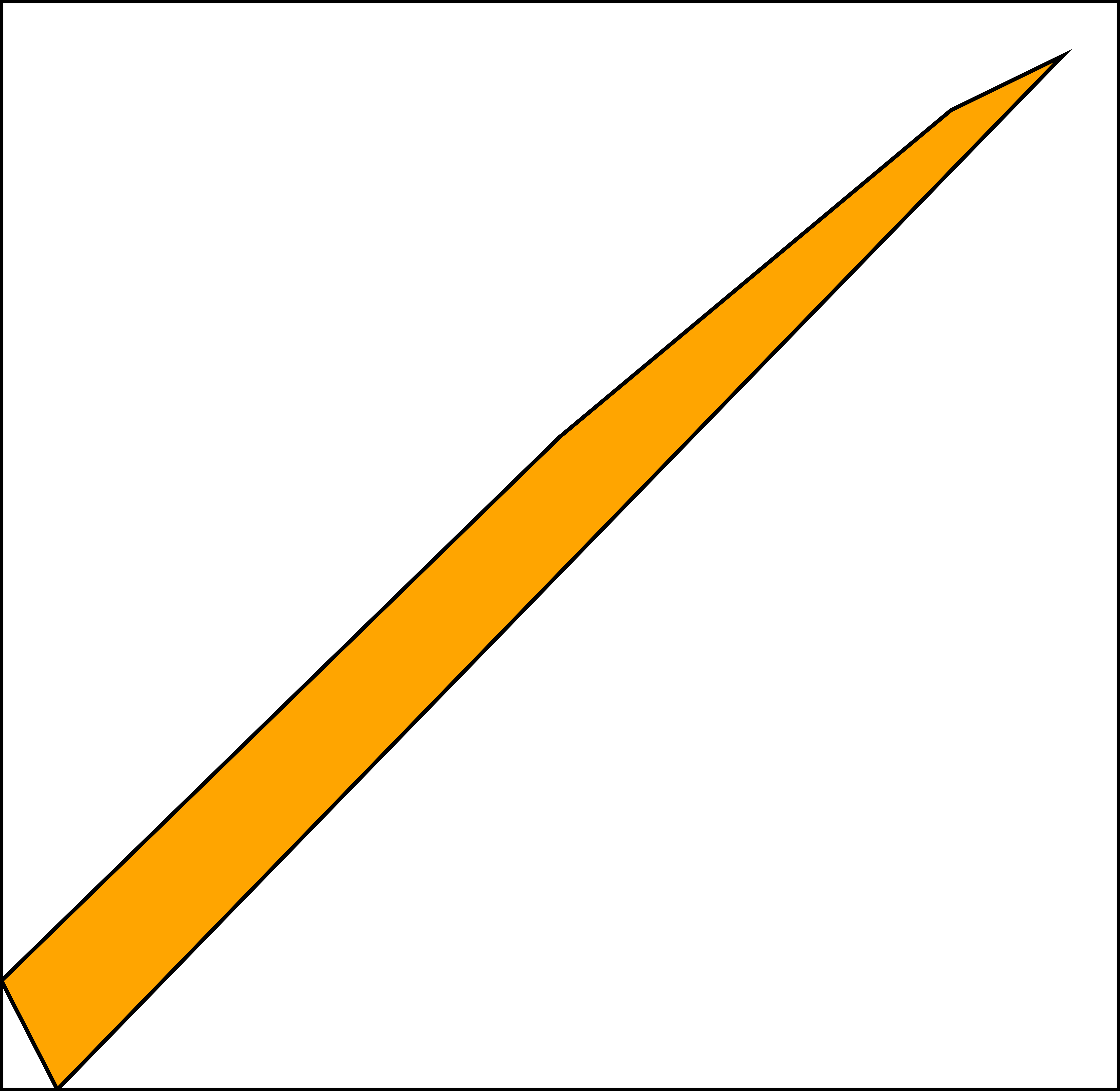}
  \caption{A ``problematic'' polygon for \textsc{Polygon Bin Packing}, which is both high and wide but has only small area.}
  \label{image_problematic_polygon}
\end{figure}

\begin{proof}
Throughout the proof we split $I$ into a constant number of sub-instances $I=\bigsqcup_i I_i$ and consider such sub-instances separately. We use without further notice that it is sufficient to show that is possible to pack each sub-instance $I_i$ into at most $\mathcal{O}(\frac{1}{\delta})\opt(I)$ bins in order to obtain an overall packing for $I$ into at most $\mathcal{O}(\frac{1}{\delta})\opt(I)$ bins.

Without loss of generality, we can assume that $\width(p) \leq \height(p)$ for all $p \in I$. All $p \in I$ with $\width(p) > \height(p)$ can be packed analogously after rotation by $\pi/2$ and the packings into unit bins can in the end be rotated back. Note that making this assumption in particular implies that $\width(p) \leq 1-\delta$ for all $p\in I$.

All polygons with $\width(p) \leq \frac{1}{2}$ can be packed with Theorem \ref{polygon_bin_packing_1/M}. So we can assume that $\width(p) \geq \frac{1}{2}$ for all $p \in I$.

Finally, we can also assume that $\area(p) < \frac{\delta}{4}$ for all $p \in I$. Indeed, all polygons $p$ with $\area(p) \geq \frac{\delta}{4}$ can be packed one by one into separate bins which needs no more than $\mathcal{O}(\frac{1}{\delta})$ bins.

Let $I_Q$ be the input consisting of $x$-parallelograms constructed from $I$ with Lemma \ref{polygon_into_parallelogram}. Assume that all of them are tilted to the right and pack those that are not tilted to the right analogously. Note that for each $q \in I_Q$,
\[
\base(q) = \frac{\area(q)}{\height(q)} \leq \frac{2\area(p)}{\height(p)} \leq \delta,
\]
where we used that $\area(p) \leq \frac{\delta}{4}$ and $\height(p)\geq \frac{1}{2}$.
Let $I_R \in \rectangles$ be the instance that for each $q \in I_Q$ contains a rectangle $r=(\base(q),\height(q))$. We now pack $I_R$ with NFDH into a strip of width $\delta$. Such strip has height at most $2 \area(I_R)/\delta + h_{\max}$ by Theorem \ref{nfdh}. Hence, using NF, we can distribute the shelves from this packing, which needs at most $4 \area(I_R)/\delta + 3$ bins of width $\delta$ and height $1$ to pack the rectangles into.

After ordering the parallelograms corresponding to a shelf in the packing of the rectangles by their angle, they can be packed into a shelf of width $1$.

Therefore, $I_Q$ can be packed into
\[
\frac{4 \area(I_R)}{\delta}+3 \leq \frac{8 \area(I)}{\delta}+3
\]
unit bins. Hence, the described method packs $I$ into at most
\[
\frac{8}{\delta}\opt(I) + 3 \leq \bigg (\frac{8}{\delta} + 3 \bigg ) \opt(I)
\]
bins.
\end{proof}

\subsection{An Efficient $\mathcal{O}(C)$-Approximation for \textsc{Polygon Bin Packing} when each Polygon has a Spine from a Set of Constant Size $C$}
\label{section_polygon_same_spine}

In this section, we show that it is possible to pack an instance $I \in \polygons$, for which all polygons $p \in I$ have a spine $s$ from a set $S$ of constant size $C \in \N$, efficiently into at most $\mathcal{O}(C) \opt(I)$ bins efficiently.

First, we need to introduce some notation.

\begin{definition}
Let $p \subset [0,1]^2$ be a polygon. We define
\[
\eta(p) := \max\{\width([x_1,x_2] \times \{y\}) \, | \, [x_1,x_2] \times \{y\} \subset p \}.
\]
If there is a spine $s$ fixed for $p$, we denote by $p^1$ the part of $p$ that lies on the left of $s$ and by $p^2$ the part of $p$ that lies on the right of $s$. In this case, we also set
\[
\eta_1(p) := \eta(p^1) \quad \text{and} \quad \eta_2(p) := \eta(p^2).
\]
\end{definition}

Note that $\max\{\eta_1(p),\eta_2(p)\} \leq \eta(p) \leq \eta_1(p)+\eta_2(p)$ for every polygon $p \in \polygons$.

We prove the following result.

\begin{theorem}
\label{thm_equal_spine}
There is an efficient $\mathcal{O}(1)$-approximation algorithm for \textsc{Polygon Bin Packing} for instances $I \in \polygons$ with the property that all polygons share a spine $s$ (up to translation) with height $H = \height(s) \geq \frac{3}{4}$.
\end{theorem}

\begin{remark}
The factor $\frac{3}{4}$ in the statement is chosen arbitrarily in $(\frac{2}{3},1]$ and only affects the exact approximation ratio. When the factor approaches $\frac{2}{3}$, the approximation guarantee approaches $\infty$.
\end{remark}

\begin{proof}
We assume that $s$ is tilted to the right. As the polygons share a spine $s$, they all have the same height $H$.

Throughout the proof, we write  $M_y:=[0,1] \times \{y\}$ for $y \in [0,1]$. Furthermore, fix some translation $\phi_L \in \Phi$ so that $s_L := \phi_L(s)$ is in the upper left corner of the bin (so that it touches the top and the left side of the bin). Similarly, fix a translation $\phi_R \in \Phi$ so that $s_R := \phi_R(s)$ is in the lower right corner of the bin (so that it touches the bottom and the right side of the bin). We define $x_{\min}$ as the $x$-coordinate of the intersection point of $s_L$ with $M_{1/2}$. Similarly, we define $x_{\max}$ as the $x$-coordinate of the intersection point of $s_R$ with $M_{1/2}$. The two numbers can be explicitly calculated to be $x_{\min} = (1- \frac{1}{2H})\width(s)$ and $x_{\max} = 1-(1- \frac{1}{2H})\width(s)$.

The idea is to give a structural result of an optimal multi-packing $\mathcal{P}_{\opt}$ of the problem. We will then construct a multi-packing, which, using this structural result, can be shown to use at most $\mathcal{O}(1)\opt(I)$ bins.

Let $P$ be an arbitrary packing of a subset $I' \subset I$ into a unit bin. We write $N := \lvert I' \rvert$. Note that since $H > \frac{1}{2}$, all polygons in the packing have their spine crossing $M_{1/2}$ at an $x$-coordinate in $[x_{\min},x_{\max}]$. Let $x_1, \dots, x_{N}$ be the $x$-coordinates of those intersection points in increasing order and denote by $p_1, \dots, p_N$ the affiliated polygons.

\begin{claim}
For all $j \in [N-1]$, it holds that $x_{j+1}-x_j \geq \frac{3H-2}{2(2H-1)}(\eta_2(p_j) + \eta_1(p_{j+1}))$.
\end{claim}

\begin{proof}
Fix some $j \in [N-1]$ and note that $\max\{\eta_2(p_j),\eta_1(p_{j+1})\} \geq \frac{1}{2}(\eta_2(p_j) + \eta_1(p_{j+1}))$. Hence it suffices to show that $x_{j+1}-x_j \geq \frac{3H-2}{2H-1} \max\{\eta_2(p_j),\eta_1(p_{j+1})\}$.

Assume that $\eta_2(p_j) \geq \eta_1(p_{j+1})$. The argument in the other case is analogous. Choose $y \in [0,1]$ so that $\width(P(p^2_j) \cap M_y) = \eta_2(p_j)$ and denote the right-most point in the line $P(p^2_j) \cap M_y$ by $z$. We note that since both $p_j$ and $p_{j+1}$ have height $H$, it surely holds that the distance between $y$ and $\pi_y(P(p_{j+1}))$, which we denote by $\dist(y, \pi_y(P(p_{j+1})))$, is bounded by $1-H$. If $\dist( y, \pi_y(P(p_{j+1}))) = 0$, then surely $x_{j+1}-x_j \geq \width(P(p^2_j) \cap M_y) = \eta_2(p_j)$.
\begin{figure}[h]
\centering
  \includegraphics[width=.5\linewidth]{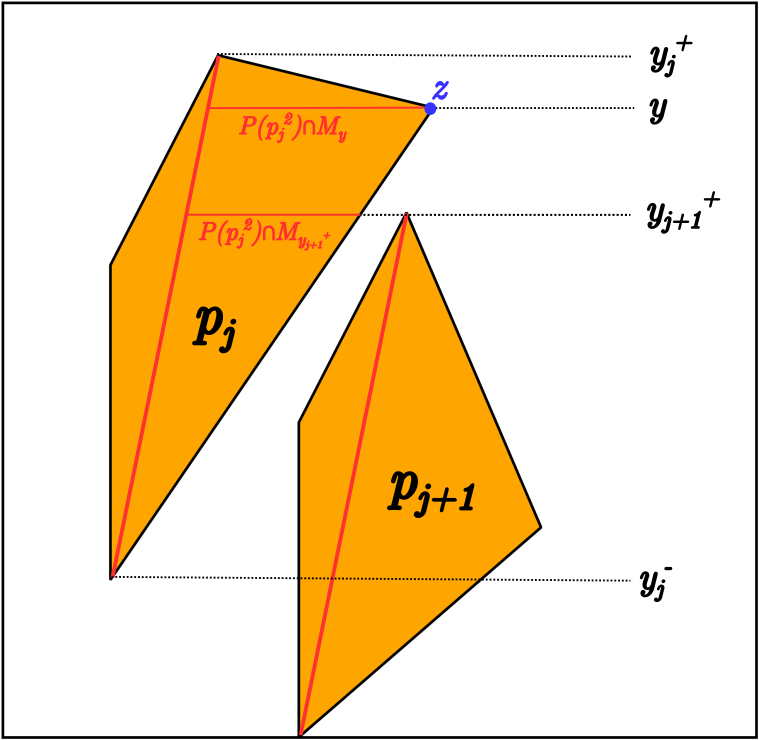}
  \caption{A schematic of the proof of $x_{j+1}-x_j \geq \frac{3H-2}{2(2H-1)}(\eta_2(p_j) + \eta_1(p_{j+1}))$.}
  \label{im_claim_polygon_eqspine}
\end{figure}

So assume that $\dist(y, \pi_y(P(p_{j+1}))) > 0$ and hence $y$ is either above or below $\pi_y(P(p_{j+1}))$. Without loss of generality, assume that it is below and therefore $y_{j+1}^+ := \max \{y' \in \pi_y(P(p_{j+1}))\} < y$ (see Figure \ref{im_claim_polygon_eqspine}). Due to convexity of $P(p_j)$, the line connecting the lower end of the spine of $P(p_j)$ (which is say at height $y_j^- \in [0,1]$) with the point $z$ lies inside of $P(p_j)$. This implies that
\begin{equation}
\label{claim_eq_1}
\width(P(p_j^2) \cap M_{y_{j+1}^+}) \geq \frac{y_{j+1}^+ - y_j^-}{y-y_j^-} \eta_2(p_j).
\end{equation}
Furthermore, since $y_j^- \leq y_{j+1}^+ < y \leq y_j^+ := \max\{y' \in \pi_y(P(p_{j}))\}$, it holds that
\begin{equation}
\label{claim_eq_2}
x_{j+1}-x_j \geq \width(P(p^2_j) \cap M_{y_{j+1}^+}).
\end{equation}
Hence, using
\[
y-y_j^- \geq y_{j+1}^+ - y_j^- = 1-\underbrace{(1-y_{j+1}^+)}_{\leq 1-H}-\underbrace{y_j^-}_{\leq 1-H} \geq 2H-1,
\]
as well as (\ref{claim_eq_1}) and (\ref{claim_eq_2}), we find that
\begin{align*}
x_{j+1}-x_j &\geq \width(P(p^2_j) \cap M_{y_{j+1}^+}) \\
&\geq \frac{y_{j+1}^+ - y_j^-}{y-y_j^-} \eta_2(p_j) \\
&= \frac{y-y_j^- -(y-y_{j+1}^+)}{y-y_j^-} \eta_2(p_j) \\
&\geq \frac{y-y_j^- -(1-H)}{y-y_j^-} \eta_2(p_j) \\
&\geq \frac{3H-2}{2H-1} \eta_2(p_j),
\end{align*}
where in the last inequality, we used that $y-y_j^- \leq y_j^+ - y_j^- \leq 1-H$. This shows the claim.
\end{proof}

From the claim it follows that
\begin{align*}
x_{\max}-x_{\min} &\geq x_N-x_1 \\
& \geq \sum_{j=1}^{N-1} (x_{j+1}-x_j) \\
& \geq \frac{3H-2}{2(2H-1)} \bigg(\sum_{j=2}^{N-1} (\eta_1(p_j) + \eta_2(p_j)) + (\eta_2(p_1) + \eta_1(p_N))\bigg) \\
& \geq \frac{3H-2}{2(2H-1)}\sum_{j=2}^{N-1} \eta(p_j) + \frac{3H-2}{2(2H-1)} (\eta_2(p_1) + \eta_1(p_N)) \\
& \geq \frac{3H-2}{2(2H-1)}\sum_{j=2}^{N-1} \eta(p_j).
\end{align*}

Now, consider an optimal multi-packing $\mathcal{P}_{\opt}=\{P_i\}_{i=1}^{\opt(I)}$ and take out the first and the last polygon (with respect to the ordering coming from $M_{1/2}$) in each bin. If there is only $1$ polygon in a bin, only take out that instead. Denote the polygons that did not get removed this way by $\overline{I} \subset I$.  We pack the removed polygons $I \backslash \overline{I}$ into $1$ bin each, which leads to a multi-packing $\overline{\mathcal{P}}=\{\overline{P}_i\}_{i=1}^{3 \opt(I)}$,where we note that some of the packings might be empty. By what we have discussed before, this packing has the property that for $i \in [\opt(I)]$, it holds that
\begin{equation}
\label{eq_upper_bounding_polygons}
x_{\max}-x_{\min} \geq \frac{3H-2}{2(2H-1)}\sum_{p \text{ is a polygon in the packing } \overline{P_i}} \eta(p)
\end{equation}
and all other packings pack at most $1$ item each. Note that $\overline{I} \subset I$ is exactly the set of polygons packed by any $\overline{P_i}$, for which $i \in [\opt(I)]$.

The idea now is to make a guess $g$ for $\lvert I \backslash \overline{I} \rvert$, remove $g$ items with high $\eta$-value and to then compute a packing for the remaining polygons into not much more than $g$ bins.

Note that surely $\lvert I \backslash \overline{I} \rvert \leq n := \lvert I \rvert$. For $g=1, \dots, n$, take out the $g$ polygons $p \in I$ with highest values $\eta(p)$ and denote the obtained instance by $I_g$.

The idea now is to partition $I_g = \bigsqcup_{k=1}^{K_g} I_{g,k}$ into as few sets as possible in such a way that for each $k \in [K_g]$, it holds that
\begin{equation}
    \label{eq_what_partition_should_do}
    x_{\max}-x_{\min} \geq 2 \sum_{p \in I_{g,k}} \eta(p),
\end{equation}
or $I_{g,k}$ contains only $1$ polygon $p$. Polygons in some $I_{g,k}$ can be packed into a single bin as follows.

Fix some $k \in \N$ and say $I_{g,k}$ consists of polygons $p_1, \dots, p_{N_{g,k}}$. If $N_{g,k}=1$, then we can clearly pack $I_{g,k}$ into a unit bin. So assume that $N_{g,k} > 1$. We construct a packing $P_{g,k}$ of $I_{g,k}$ into a strip $[0,\infty) \times [0,1]$ as follows. Place polygon $p_1$ so that it touches the left and the upper boundary of the strip. Then place polygon $p_2$ below it so that it touches the left boundary of the strip, as well as $p_1$. Further polygons $p_3, p_4, \dots$ are now being placed so that a new placed polygon always touches the left boundary of the strip and the last placed polygon. Repeat this process until there is not enough space due to the bottom of the strip anymore. We do now place a next polygon so that it touches the the last placed polygon and the bottom of the strip instead. The remaining polygons are now being placed so that a newly placed polygon always touches the bottom boundary of the strip and the last placed polygon. Note that the fact that all polygons have the same height $H$ ensures that the procedure works, and no polygon touches another one except for the one placed directly before and after it. In the following, we show that $\width(P_{g,k}) \leq 1$ and that hence $P_{g,k}$ fits into a unit bin.

To see this, write $x_1, \dots, x_{N_{g,k}}$ for the $x$-coordinates of the intersection points of the spines of $p_1, \dots, p_{N_{g,k}}$ in $P_{g,k}$ with $M_{1/2}$. First, consider $p_1$. Note that $P_{g,k}(p_1)$ necessarily intersects with $s_L$ (which we recall is the placement of the spine $s$ so that it is in the upper left corner). This is because $P_{g,k}(p_1)$ touches the left boundary, contains a copy of this spine and has its lower end at the same height as the lower end of $s_L$. Thus, if $y_1 \in [1-H,1]$ is any height at which such intersection happens, we denote by $t$ the segment at height $y_1$ that is parallel to the $x$-axis, starts on $s_L$ and ends at the spine of the placement of $p_1$. We find that
\[
x_1 - x_{\min}(I) = \width(t) \leq \eta_1(p_1) \leq \eta(p_1).
\]

Now let $j \in [N_{g,k}-1]$ and consider $p_j$ and $p_{j+1}$. Then they touch somewhere in $P_{g,k}$, say at height $y_j \in [0,1]$. Therefore
\begin{align*}
x_{j+1} - x_j &\leq \width((p^2_j \cup p^1_{j+1}) \cap M_{y_j}) \leq \eta_2(p_j) + \eta_1(p_{j+1}) \leq \eta(p_j) + \eta(p_{j+1}).
\end{align*}

Writing $x_0 := x_{\min}$ and using (\ref{eq_what_partition_should_do}), we hence find that
\[
x_{N_{g,k}} = \sum_{j=0}^{N_{g,k}-1} (x_{j+1} - x_j) \leq 2 \sum_{p \in I_{g,k}} \eta(p) + x_{\min}(I) - \eta(p_{N_{g,k}})  \leq x_{\max}(I) - \eta(p_{N_{g,k}}).
\]
We finally consider $p_{N_{g,k}}$. Note that since $p_{N_{g,k}}$ is placed on the bottom of the strip, its spine in $P_{g,k}$ can at most extend horizontally to the right until $x_{N_{g,k}} + (1-\frac{1}{2H})\width{s}$ (that is, the upper end of the spine of $p_{N_{g,k}}$ has such $x$-coordinate). But this is exactly equal to $x_{N_{g,k}} + (1-x_{\max}(I))$. In particular,
\[
\width(P_{g,k}) \leq x_{N_{g,k}} + (1-x_{\max}(I)) + \eta_2(p_{N_{g,k}}) \leq 1 + x_{N_{g,k}} -(x_{\max}(I) - \eta(p_{N_{g,k}})) \leq 1,
\]
which shows that $I_{g,k}$ can be packed by $P_{g,k}$ into a unit bin.

So it suffices to find a partition $I_g = \bigsqcup_{k=1}^{K_g} I_{g,k}$ satisfying (\ref{eq_what_partition_should_do}) for $g \in [2n]$ and show that there is $g$ so that $K_g + g$ is not much bigger than $\opt(I)$. 

We first show the following claim.

\begin{claim}
Let $I_{BP}$ be an instance for $(1D)$-\textsc{Bin Packing} and let $\alpha \geq 1$. Define $\overline{I_{BP}} := \{\overline{s}\}_{s \in I_{BP}}$, where for $s \in I_{BP}$,  we set $\overline{s}:=\min\{\alpha s,1\}$. Then
\[
\opt_{BP}(\overline{I_{BP}}) \leq (2 \alpha+1) \opt_{BP}(I_{BP}),
\]
where $opt_{BP}$ is the function that assigns to an instance of \textsc{Bin Packing} its optimal solution.
\end{claim}
\begin{proof}
Consider the partition $I_{BP}=\bigsqcup_{j=1}^l A_j$ corresponding to an optimal solution. Then $\sum_{s \in A_j} s \leq 1$ and hence, writing $\overline{A_j} = \{\overline{s}\}_{s \in A_j}$, it holds that $\sum_{\overline{s} \in \overline{A_j}} \overline{s} \leq \alpha$. With \textsc{NextFitDecreasing}, we can hence pack $\overline{A_j}$ into at most $2 \alpha + 1$ bins . Doing this for all $j \in [l]$, we obtain a feasible solution of \textsc{Bin Packing} for the instance $\overline{I_{BP}}$ into at most $(2\alpha + 1) l = (2 \alpha + 1) \opt_{BP}(I_{BP})$ bins.
\end{proof}

For $p \in I_g$, we write $\rho(p):=\min \{2\eta(p),x_{\max}-x_{\min}\}$. Use FFDH to find a $\frac{3}{2}$-approximation to \textsc{Bin Packing} efficiently (see~\cite{SimchiLevi1994NewWR}) for the instance $\{\rho(p)\}_{p \in I_g}$ but for bins of size $x_{\max}-x_{\min}$. Denote the number of bins this solution needs by $K_g$ and partition $I_g$ according to this solution into sets $I_g = \bigsqcup_{k=1}^{K_g} I_{g,k}$. Note that the partition satisfies exactly the property (\ref{eq_what_partition_should_do}).

Denote the optimal solution of an instance $J$ for \textsc{Bin Packing} for bins of size $x_{\max}-x_{\min}$ by $\opt_{BP^s}(J)$. We now consider how big the value of $K_g$ is in the case when $g=\lvert I \backslash \overline{I} \rvert$. In this case, using the last claim for $\alpha := 2/\big(\frac{3H-2}{2(2H-1)}\big) = \frac{4(2H-1)}{3H-2}$ and recalling that $I_g$ is obtained from $I$ after taking out the $g$ polygons $p$ with biggest values of $\eta(p)$, we find that
\begin{align*}
K_{\lvert I \backslash \overline{I} \rvert} & \leq \frac{3}{2}\opt_{BP^s}(\{\rho(p)\}_{p \in I_{\lvert I \backslash \overline{I} \rvert}}) \\
&\leq \frac{3}{2}\opt_{BP^s}(\{\rho(p))\}_{p \in \overline{I}}) \\ 
&\leq \frac{3}{2} \bigg(\frac{8(2H-1)}{3H-2}+1 \bigg) \opt_{BP^s} \bigg( \Big\{\frac{3H-2}{2(2H-1)} \eta(p) \Big\}_{p \in \overline{I}}\bigg) \\
&\leq \frac{3}{2} \bigg(\frac{8(2H-1)}{3H-2}+1 \bigg) \opt(I) \\
&= \bigg(\frac{24H-12}{3H-2} + \frac{3}{2}\bigg) \opt(I),
\end{align*}
where the last inequality follows from (\ref{eq_upper_bounding_polygons}).

So take the value of $g \in [n]$ for which $g+K_g$ is minimized. Then, recalling that $\lvert I \backslash \overline{I} \rvert \leq 2 \opt(I)$, it holds that
\[
g+K_g \leq \lvert I \backslash \overline{I} \rvert + K_{\lvert I \backslash \overline{I} \rvert} \leq \bigg(\frac{7}{2} + \frac{24H-12}{3H-2}\bigg) \opt(I) \leq 27.5 \opt(I),
\]
using that $H \geq \frac{3}{4}$. This concludes the proof.
\end{proof}

\begin{remark}
A next goal could be to show the existence of an efficient $\mathcal{O}(1)$-approximation for instances $I \in \polygons$ for which there is an angle $\alpha \in [0, \pi]$ so that each polygon $p \in I$ has a spine $s_p$ with angle $\alpha$.
\end{remark}

We conclude the section with the following corollary.

\begin{corollary}
Let $C \in \N$. There is an efficient $\mathcal{O}(C)$-approximation algorithm for \textsc{Polygon Bin Packing} for instances $I \subset \polygons$ satisfying the following. There is a set $S$ of lines in $[0,1]^2$, so that $\lvert S \rvert \leq C$ and with the property that each polygon $p \in I$ has a spine $s$ in the set $S$ (up to translation).
\end{corollary}

\begin{proof}
First, we note that all polygons with height smaller than $\frac{3}{4}$ can easily be packed using Theorem \ref{polygon_bin_packing_1_minus_delta} for $\delta = \frac{1}{4}$. This needs at most $\mathcal{O}(1) \opt(I)$ bins. We call the remaining set of polygons $\overline{I}$. Partition $\overline{I}:=\bigsqcup_{s \in S} I_s$ so that for all $s \in S$, all $p \in I_s$ have a spine that (up to a translation) is equal to $s$. Packing each $I_s$, $s \in S$, separately with the algorithm from Theorem \ref{thm_equal_spine}, we need at most $C\mathcal{O}(1)\opt(\overline{I}) \leq \mathcal{O}(C)\opt(I)$ bins.
\end{proof}

\bibliography{ref_MT.bib}

\begin{thebibliography}{10}

\bibitem{aamand_23}
Anders Aamand, Mikkel Abrahamsen, Lorenzo Beretta, and Linda Kleist.
\newblock Online sorting and translational packing of convex polygons.
\newblock In {\em Proceedings of the 2023 Annual ACM-SIAM Symposium on Discrete
  Algorithms (SODA)}, pages 1806--1833, 2023.
\newblock \href {https://doi.org/10.1137/1.9781611977554.ch69}
  {\path{doi:10.1137/1.9781611977554.ch69}}.

\bibitem{alt_17}
Helmut Alt, Mark de~Berg, and Christian Knauer.
\newblock Approximating minimum-area rectangular and convex containers for
  packing convex polygons.
\newblock {\em Journal of Computational Geometry}, 8(1):1--10, 2017.
\newblock \href {https://doi.org/10.20382/jocg.v8i1a1}
  {\path{doi:10.20382/jocg.v8i1a1}}.

\bibitem{alt_17_corrigendum}
Helmut Alt, Mark de~Berg, and Christian Knauer.
\newblock Corrigendum to: Approximating minimum-area rectangular and convex
  containers for packing convex polygons.
\newblock {\em Journal of Computational Geometry}, 11(1):653--655, 2020.
\newblock \href {https://doi.org/10.20382/jocg.v11i1a26}
  {\path{doi:10.20382/jocg.v11i1a26}}.

\bibitem{bansal_06}
Nikhil Bansal, Jos\'{e}~R. Correa, Claire Kenyon, and Maxim Sviridenko.
\newblock Bin packing in multiple dimensions: Inapproximability results and
  approximation schemes.
\newblock {\em Mathematics of Operations Research}, 31(1):31--49, 2006.
\newblock \href {https://doi.org/10.1287/moor.1050.0168}
  {\path{doi:10.1287/moor.1050.0168}}.

\bibitem{khanbansal}
Nikhil Bansal and Arindam Khan.
\newblock Improved approximation algorithm for two-dimensional bin packing.
\newblock In {\em Proceedings of the 2014 Annual ACM-SIAM Symposium on Discrete
  Algorithms (SODA)}, page 13–25, 2014.
\newblock \href {https://doi.org/10.1137/1.9781611973402.2}
  {\path{doi:10.1137/1.9781611973402.2}}.

\bibitem{coffman_strip}
E.~G. Coffman, Jr., M.~R. Garey, D.~S. Johnson, and R.~E. Tarjan.
\newblock Performance bounds for level-oriented two-dimensional packing
  algorithms.
\newblock {\em SIAM Journal on Computing}, 9(4):808--826, 1980.
\newblock \href {https://doi.org/10.1137/0209062} {\path{doi:10.1137/0209062}}.

\bibitem{dosa_sgall}
Gy{\"o}rgy D{\'o}sa and Jiří Sgall.
\newblock {First Fit bin packing: A tight analysis}.
\newblock In Natacha Portier and Thomas Wilke, editors, {\em 30th International
  Symposium on Theoretical Aspects of Computer Science (STACS 2013)}, pages
  538--549, Dagstuhl, Germany, 2013. Schloss Dagstuhl--Leibniz-Zentrum fuer
  Informatik.
\newblock \href {https://doi.org/10.4230/LIPIcs.STACS.2013.538}
  {\path{doi:10.4230/LIPIcs.STACS.2013.538}}.

\bibitem{dosa_sgall_bf}
György Dósa and Jiří Sgall.
\newblock Optimal analysis of best fit bin packing.
\newblock In Javier Esparza, Pierre Fraigniaud, Thore Husfeldt, and Elias
  Koutsoupias, editors, {\em Automata, Languages, and Programming}, pages
  429--441, Berlin, Heidelberg, 2014. Springer Berlin Heidelberg.
\newblock \href {https://doi.org/10.1007/978-3-662-43948-7_36}
  {\path{doi:10.1007/978-3-662-43948-7_36}}.

\bibitem{harren}
Rolf Harren, Klaus Jansen, Lars Pr{\"a}del, and Rob van Stee.
\newblock A (5/3{\thinspace}+{\thinspace}$\epsilon$)-approximation for strip
  packing.
\newblock In Frank Dehne, John Iacono, and J{\"o}rg-R{\"u}diger Sack, editors,
  {\em Algorithms and Data Structures}, pages 475--487, Berlin, Heidelberg,
  2011. Springer Berlin Heidelberg.
\newblock \href {https://doi.org/10.1007/978-3-642-22300-6_40}
  {\path{doi:10.1007/978-3-642-22300-6_40}}.

\bibitem{harren_vanstee}
Rolf Harren and Rob van Stee.
\newblock An absolute 2-approximation algorithm for two-dimensional bin
  packing, 2009.
\newblock \href {https://doi.org/10.48550/ARXIV.0903.2265}
  {\path{doi:10.48550/ARXIV.0903.2265}}.

\bibitem{hoberg}
Rebecca Hoberg and Thomas Rothvoss.
\newblock A logarithmic additive integrality gap for bin packing.
\newblock In {\em Proceedings of the 2017 Annual ACM-SIAM Symposium on Discrete
  Algorithms (SODA)}, pages 2616--2625, 2017.
\newblock \href {https://doi.org/10.1137/1.9781611974782.172}
  {\path{doi:10.1137/1.9781611974782.172}}.

\bibitem{JANSEN_strip}
Klaus Jansen and Roberto Solis-Oba.
\newblock Rectangle packing with one-dimensional resource augmentation.
\newblock {\em Discrete Optimization}, 6(3):310--323, 2009.
\newblock \href {https://doi.org/https://doi.org/10.1016/j.disopt.2009.04.001}
  {\path{doi:https://doi.org/10.1016/j.disopt.2009.04.001}}.

\bibitem{johnson_1974}
D.~S. Johnson, A.~Demers, J.~D. Ullman, M.~R. Garey, and R.~L. Graham.
\newblock Worst-case performance bounds for simple one-dimensional packing
  algorithms.
\newblock {\em SIAM Journal on Computing}, 3(4):299--325, 1974.
\newblock \href {https://doi.org/10.1137/0203025} {\path{doi:10.1137/0203025}}.

\bibitem{johnson_phd}
D.S. Johnson.
\newblock {\em Near-optimal bin packing algorithms}.
\newblock PhD thesis, Massachusetts Institute of Technology, 1973.

\bibitem{leung}
Joseph Y-T. Leung, Tommy~W. Tam, C.S. Wong, Gilbert~H. Young, and Francis~Y.L.
  Chin.
\newblock Packing squares into a square.
\newblock {\em Journal of Parallel and Distributed Computing}, 10(3):271--275,
  1990.
\newblock \href {https://doi.org/https://doi.org/10.1016/0743-7315(90)90019-L}
  {\path{doi:https://doi.org/10.1016/0743-7315(90)90019-L}}.

\bibitem{SimchiLevi1994NewWR}
David Simchi-Levi.
\newblock New worst‐case results for the bin‐packing problem.
\newblock {\em Naval Research Logistics}, 41:579--585, 1994.
\newblock \href
  {https://doi.org/10.1002/1520-6750(199406)41:4<579::AID-NAV3220410409>3.0.CO;2-G}
  {\path{doi:10.1002/1520-6750(199406)41:4<579::AID-NAV3220410409>3.0.CO;2-G}}.

\end{thebibliography}

\end{document}